\ifx\CCCG\undefined
\newcommand{\ConfVer}[1]{}
\newcommand{\FullVer}[1]{#1}
\else
\newcommand{\ConfVer}[1]{#1}
\newcommand{\FullVer}[1]{}
\fi

\FullVer{
\documentclass[11pt]{article}%
\usepackage{wide}
}
\ConfVer{
\documentclass{cccg13}
}

\usepackage{ntheorem}
\usepackage{salgorithm}%
\usepackage[breaklinks]{hyperref} 
\usepackage{color}%
\usepackage{xspace}%
\usepackage{euscript}%
\usepackage{amstext}
\usepackage{paralist}%
\usepackage{amsmath}%
\usepackage{amssymb}%
\usepackage{graphicx}
\FullVer{
\usepackage{theorem}%
}
\definecolor{red25}{rgb}{0.4,0,0.0}%
\newcommand{\PStyle}[1]{\textcolor{red25}{\textsc{#1}}}

\newcommand{\FTC}{\PStyle{FTC}\xspace}
\newcommand{\FTM}{\PStyle{FTM}\xspace}
\newcommand{\kCenter}{\mathcal{A}_{\mathsf{c}}}
\newcommand{\kMedian}{\mathcal{A}_{\mathsf{m}}}
\newcommand{\fCenter}{\mathcal{A}_{\mathsf{fc}}}
\newcommand{\fMedian}{\mathcal{A}_{\mathsf{fm}}}

\definecolor{blue25}{rgb}{0,0,0.7}
\newcommand{\emphic}[2]{%
     \textcolor{blue25}{%
         \textbf{\emph{#1}}}%
         \index{#2}}
\newcommand{\emphi}[1]{\emphic{#1}{#1}}

\newcommand{\obslab}[1]{\label{observation:#1}}
\newcommand{\obsref}[1]{Observation~\ref{observation:#1}}

\providecommand{\lemlab}[1]{\label{lemma:#1}}
\providecommand{\lemref}[1]{Lemma~\ref{lemma:#1}}

\newcommand{\myqedsymbol}{\rule{2mm}{2mm}}
\FullVer{
  \newtheorem{theorem}{Theorem}[section]
  \newtheorem{lemma}[theorem]{Lemma}
  \newenvironment{proof}{\trivlist\item[]\emph{Proof}:}%
                  {\unskip\nobreak\hskip 1em plus 1fil\nobreak%
                           \myqedsymbol
                           \parfillskip=0pt%
                           \endtrivlist}
}
  \newenvironment{proofof}[1]{\smallskip\noindent{\bf Proof of #1:}}%
  {\unskip\nobreak\hskip 1em plus 1fil\nobreak%
     \myqedsymbol
     \parfillskip=0pt%
     \endtrivlist}
\ConfVer{
  \renewtheorem{theorem}{Theorem}[section]
  \renewtheorem{lemma}[theorem]{Lemma}
}
\newtheorem{claim}[theorem]{Claim}

\newtheorem{observation}[theorem]{Observation}

\newtheorem{problem}[theorem]{Problem}
\newtheorem{corollary}[theorem]{Corollary}
\newcommand{\Term}[1]{\textsf{#1}}
\newcommand{\TermI}[1]{\Term{#1}\index{#1@\Term{#1}}}

\newcommand{\PTAS}{\TermI{PTAS}\xspace}

\theoremstyle{remark}{\theorembodyfont{\rm}
}%

\newcommand{\thmlab}[1]{{\label{theo:#1}}}
\newcommand{\thmref}[1]{Theorem~\ref{theo:#1}}

\newcommand{\eqlab}[1]{\label{equation:#1}}
\newcommand{\Eqref}[1]{Eq.~(\ref{equation:#1})}
\newcommand{\seclab}[1]{{\label{section:#1}}}
\newcommand{\secref}[1]{Section~\ref{section:#1}}

\newcommand{\claimlab}[1]{\label{claim:#1}}
\newcommand{\claimref}[1]{Claim~\ref{claim:#1}}

\newcommand{\brc}[1]{\left\{ {#1} \right\}}

\newcommand{\pth}[2][\!]{#1\left({#2}\right)}

\newcommand{\dist}[2]{\mathsf{d}\pth{#1,#2}}
\newcommand{\distPk}[3]{\mathsf{d}_{#3}\pth{#2,#1}}
\newcommand{\nnPk}[3]{\mathsf{nn}_{#3}\pth{#2,#1}}

\newcommand{\NNPk}[3]{\mathsf{NN}_{#3}\pth{#2,#1}}

\newcommand{\floor}[1]{\left\lfloor {#1} \right\rfloor}

\newcommand{\cardin}[1]{\left\lvert {#1} \right\rvert}

\newcommand{\eps}{{\varepsilon}}%
\newcommand{\divides}{|}

\newcommand{\etal}{\textit{et~al.}\xspace}

\newcommand{\PntSet}{\mathsf{P}}
\newcommand{\PntSetA}{Q}
\newcommand{\PntSetB}{S}
\newcommand{\PntSetC}{W}
\newcommand{\PntSetX}{\mathsf{X}}
\newcommand{\PntSetY}{\mathsf{Y}}
\newcommand{\ccost}[2]{\mu\pth{#1,#2}}

\newcommand{\const}{c}
\newcommand{\qtnt}{m}

\newcommand{\setC}{C}

\newcommand{\pnt} {\mathsf{p}}
\newcommand{\pntA}{\mathsf{q}}
\newcommand{\pntB}{\mathsf{v}}
\newcommand{\pntC}{\mathsf{w}}

\newcommand{\pntY}{\mathsf{y}}

\newcommand{\remove}[1]{}

\newcommand{\mtrA}{M}
\newcommand{\COPT}{C^{*}}
\newcommand{\ropt}{r_{\mathsf{opt}}}
\newcommand{\ralg}{r_{\mathsf{alg}}}
\newcommand{\rcen}{r_{\mathsf{cen}}}
\newcommand{\sopt}{\sigma_{\mathsf{opt}}}
\newcommand{\salg}{\sigma_{\mathsf{alg}}}
\newcommand{\smed}{\sigma_{\mathsf{med}}}
\newcommand{\ballA}{\mathsf{b}}
\newcommand{\num}{x}
\providecommand{\si}[1]{#1}

\newcommand{\ball}[2]{\mathsf{ball}\pth[]{#1,#2}}

\newcommand{\NirmanThanks}[1]{\thanks{
      University of Illinois; 
      {\tt \si{nkumar5}\atgen{}illinois.edu}; {\tt
         \url{http://www.cs.uiuc.edu/\string~\si{nkumar5}/}.} #1}}
\newcommand{\BenThanks}[1]{\thanks{
      University of Illinois; 
      {\tt \si{raichel2}\atgen{}illinois.edu}; {\tt
         \url{http://www.cs.uiuc.edu/\string~\si{raichel2}/}.} #1}}

\newcommand{\atgen}{\symbol{'100}}

\begin{document}

\title{Fault Tolerant Clustering Revisited}

\author{%
   Nirman Kumar\NirmanThanks{}
   \and%
   Benjamin Raichel \BenThanks{}}

\date{}

\maketitle

\begin{abstract}
In discrete $k$-center and $k$-median clustering, we are given a set of points 
$\PntSet$ in a metric space $\mtrA$, and the task is to output a set 
$\setC \subseteq \PntSet$, $\cardin{\setC} = k$, such that the cost of clustering $\PntSet$
using $\setC$ is as small as possible. For $k$-center, the cost is the furthest a point has to 
travel to its nearest center, whereas for $k$-median, 
the cost is the sum of all point to nearest 
center distances.  In the fault-tolerant versions of these problems, we are given an additional
parameter $1 \leq \ell \leq k$, such that when computing the cost of clustering, 
points are assigned to their 
$\ell$th nearest-neighbor in $\setC$, instead of their nearest neighbor.  
We provide constant factor 
approximation algorithms for these problems that are both conceptually simple and 
highly practical from an implementation stand-point.
\end{abstract}

\section{Introduction}
Two of the most common clustering problems
are $k$-center and $k$-median clustering. In both these problems, the 
goal is to find the minimum cost partition of a given point set $\PntSet$ 
into $k$ clusters.
Each cluster is defined by a point in the set of cluster \emphi{centers}, 
$\setC \subseteq \PntSet$, where $\cardin{\setC} = k$. 
In $k$-center clustering, the 
cost is the maximum distance of a point to its assigned cluster center, 
and in $k$-median clustering,
the cost is the sum of distances of points to their assigned cluster center. In both
cases, given a set of cluster centers $\setC$, a point is assigned
to its closest center in $\setC$.  
Both these problems are NP-hard for most metric spaces.
Hochbaum and Shmoys showed that $k$-center clustering has a $2$-approximation
algorithm, but for every $\eps > 0$ it cannot be approximated to better than $(2-\eps)$
unless P=NP \cite{hs-bphkc-85}. A $2$-approximation was also provided by 
Gonzalez \cite{g-cmmid-85}, and by Feder and Greene \cite{fg-oafac-88}.
For $k$-median, the best known approximation factor is $1 + \sqrt{3} + \eps$. This 
is a recent result of Li and Svensson \cite{ls-akmpa-13}, but
the approximation version of the $k$-median problem has a long history, and before the result
of Li and Svensson, the best known result was by Arya \etal \cite{agkmp-lshkm-01}, that
achieved an approximation factor of $(3+\eps)$ for any $\eps > 0$, using local search.
In general metric spaces, $k$-median is also APX hard. Jain \etal
showed that $k$-median is hard to approximate within a factor of 
$1+2/e \approx 1.736$ \cite{jms-ngflp-02}.
In Euclidean spaces, the $k$-center problem remains APX-hard \cite{fg-oafac-88}, 
while $k$-median admits a \PTAS \cite{arr-asekm-98,kr-nltas-99,hm-ckmkm-04}.

\paragraph{Fault-Tolerance.}
  As mentioned earlier, in both the $k$-center and $k$-median problems, each point is assigned to its
  closest center. Consider a realistic scenario where $k$-center clustering is used
  to decide in which $k$ of $n$ cities, certain facilities (say Sprawlmarts or hospitals) 
  are opened, so that for clients
  in the $n$ cities, their maximum distance to a facility is minimized. Once the $k$ cities
  are decided upon, clearly each client goes to its nearest such facility when it requires
  service. Due to facility downtimes however, sometimes
  clients may need to go to their second closest, or third closest facility. Thus,
  in the fault-tolerant version of the $k$-center problem, we say that the cost
  of a client is the distance to its $\ell$th nearest facility for some fixed 
  $1 \leq \ell \leq k$. 
  The problem now is to find a set of $k$ centers so that the worst case cost is minimized, 
  where in the worst case each client actually goes to its $\ell$th nearest facility, 
  and the cost of clustering is the maximum distance traveled by any client.

  The fault-tolerant $k$-center 
  problem was first studied by
  Krumke \cite{k-gpcp-95},
  who gave a $4$-approximation algorithm for this problem. 
  Chaudhuri \etal provided a 
  $2$-approximation algorithm for this problem \cite{cgr-pnkcp-98}, which is the best possible
  under standard complexity theoretic assumptions. 
  In both these papers, the version considered, differs slightly from ours in that one only 
  considers points which are not centers when computing 
  the point that has the furthest distance to its $\ell$th closest center.
  Khuller \etal \cite{kps-ftkcp-00} later considered both versions of the $k$-center problem. 
  Their first version is the same as ours, i.e. the cost
  is the maximum distance of any point (including centers) to its $\ell$th nearest center.
  They gave a $2$-approximation when $\ell < 4$ and a $3$-approximation otherwise. 
  Their second version is the same as that of Krumke \cite{k-gpcp-95}. 
  For this version, they provided a $2$-approximation 
  algorithm matching the result of Chaudhuri \etal \cite{cgr-pnkcp-98}. 

  For $k$-median clustering, a fault-tolerant version has been considered by
  Swamy and Shmoys \cite{ss-ftfl-03}. 
  In their version, $k$-centers need to be opened,
  and in addition there is a fault-tolerance parameter $r \leq k$. The cost
  for a client is the sum of distances to its $r$ closest facilities. 
  Swamy and Shmoys actually considered a much more general setting for the fault tolerant
  facility location problem, where the requirement $r_j$ for a client $j$ could
  be non-uniform. However, for the fault tolerant $k$-median problem, the
  algorithm they provided was for a uniform requirement $r_j = r$ for all clients. 
  For this problem, they provided a $4$-approximation algorithm. The fault tolerant version we 
  consider is different from the version of Swamy and Shmoys. In our version, the cost
  for a client is its distance to its $\ell$th nearest facility (instead of the sum to its 
  $l$ nearest facilities), and we add the cost for all the clients to get the cost of the clustering.

\paragraph{Our Contribution.}
  Our main contribution is in providing and proving the correctness of a 
  natural technique for fault-tolerant clustering.
  In particular, letting $\qtnt = \floor{k/\ell}$, 
  we show that given a set of centers which is a constant factor 
  approximation to the optimal $\qtnt$-center (resp. $\qtnt$-median) clustering, 
  one can easily compute a set of $k$ centers 
  whose cost is a constant factor approximation to the optimal fault-tolerant $k$-center 
  (resp. $k$-median) clustering. 
  Specifically, in order to turn the non-fault-tolerant solution into a fault-tolerant one, 
  simply add for each point of the $\qtnt$ center set, its $\ell$ nearest neighbors 
  in $\PntSet$.  In other words, our main contribution is in proving a relationship 
  between the fault-tolerant and non-fault-tolerant cases, specifically that the 
  non-fault-tolerant solution for $\qtnt$ centers is already a near optimal fault-tolerant 
  solution in that, up to a constant factor, it is enough to ``reinforce'' the current 
  center locations rather than looking for new ones. 

  For fault-tolerant $k$-center we prove that if one applies 
  this post-processing technique to any $\const$-approximate solution to the non-fault-tolerant 
  problem with $\qtnt$ centers, then one is guaranteed a $(1+2\const)$-approximation to the 
  optimal fault-tolerant clustering.  Similarly, for fault-tolerant $k$-median we show this post 
  processing technique leads to a $(1+4\const)$-approximation.

  Our second main result is that using the algorithm of Gonzalez \cite{g-cmmid-85} 
  for the initial $\qtnt$-center solution, gives a tighter approximation ratio guarantee.  
  Specifically, we get a $3$-approximation when $\ell \divides k$, 
  and a $4$-approximation otherwise, for fault-tolerant $k$-center. 
  Additionally, on the median side, to the best of our knowledge, 
  we are the first to consider this particular variant of fault-tolerant 
  $k$-median clustering.

  The approximation ratios of our algorithms are reasonable but not optimal. However, 
  the authors feel that the algorithms more than make up for this in their 
  conceptual simplicity and practicality from an implementation stand-point.  
  Notably, if one has an existing implementation of an  
  $m$-center or an $m$-median clustering approximation algorithm, 
  one can immediately turn it into a fault-tolerant clustering algorithm for $k$ centers 
  with this technique.

\paragraph{Organization.}
  In \secref{prelims} we set up notation and formally define our variants for the fault-tolerant
  $k$-center and $k$-median problems. In \secref{algos} we review the algorithm of 
  Gonzalez \cite{g-cmmid-85}, and present our algorithms for the fault-tolerant 
  $k$-center and $k$-median problems. In \secref{results} we analyze the approximation ratios 
  of our algorithm. We conclude
  in \secref{conclusions}. 


\section{Preliminaries}
\seclab{prelims}

We are given a set of $n$ points $\PntSet = \brc{\pnt_1, \dots, \pnt_n}$ in
a metric space $\mtrA$.
Let $\dist{\pnt}{\pnt'}$ denote the distance between the points $\pnt$ and $\pnt'$ in $\mtrA$.
For a point $\pnt \in \mtrA$, and a number $x \geq 0$, let $\ball{\pnt}{x}$ denote the closed
ball of radius $x$ with center $\pnt$. 
For a point $\pnt \in \mtrA$,
a subset $\PntSetB \subseteq \PntSet$, and an integer $1\leq i \leq \cardin{\PntSetB}$, let 
$\distPk{\PntSetB}{\pnt}{i}$ denote the radius of the smallest (closed) ball with center
$\pnt$ that contains at least $i$ points in the set $\PntSetB$. 
Let $\nnPk{\PntSetB}{\pnt}{i}$ denote the $i$th nearest neighbor of $\pnt$ in $\PntSetB$, 
i.e. the point in $\PntSetB$ such that
$\dist{\pnt}{\nnPk{\PntSetB}{\pnt}{i}} = \distPk{\PntSetB}{\pnt}{i}$.\footnote{In case 
of non unique distances, we use the standard technique of lexicographic ordering of
the pairs $(\dist{\pnt}{\pnt_j},j)$ to ensure that the $1$st, $2$nd, $\dots$, 
$\cardin{\PntSetB}$th, nearest-neighbors of $\pnt$ are all unique.}
Let $\NNPk{\PntSetB}{\pnt}{i} = \cup_{j=1}^i \{\nnPk{\PntSetB}{\pnt}{j}\}$ be the set of 
$i$ nearest neighbors of $\pnt$ in $\PntSetB$. By definition, for $1 \leq i \leq \cardin{S}$,
$\cardin{\NNPk{\PntSetB}{\pnt}{i}} = i$. 
The following is an easy observation.
\begin{observation}
  \obslab{lipschitz}
  For any fixed $\PntSetA \subseteq \PntSet$ and integer 
  $1 \leq i \leq \cardin{\PntSetA}$, the function $\distPk{\PntSetA}{\cdot}{i}$ is a 
  $1$-Lipschitz function of its argument, i.e., for any $\pnt, \pntA \in \mtrA$, 
  $\distPk{\PntSetA}{\pnt}{i} \leq \distPk{\PntSetA}{\pntA}{i} + \dist{\pnt}{\pntA}$.
\end{observation}
\subsection{Problem Definitions}
\begin{problem}[Fault-tolerant $k$-center]
Let $\PntSet$ be a set of $n$ points in $\mtrA$, and let $k$ and $\ell$ be two given 
integer parameters such that $1\leq \ell\leq k\leq n$.  
For a subset $\setC \subseteq \PntSet$, we define the cost function $\ccost{\PntSet}{\setC}$ as,
\begin{align*}
    \ccost{\PntSet}{\setC} = \max_{\pnt \in \PntSet} \distPk{\setC}{\pnt}{\ell}.
\end{align*}
The fault-tolerant $k$-center problem, denoted $\FTC(\PntSet,k,\ell)$, is to compute a set 
$\COPT$ with $\cardin{\COPT} = k$ such that,
\begin{align*}
    \ccost{\PntSet}{\COPT} = \min_{\setC \subseteq \PntSet, \cardin{\setC} = k}
                             \ccost{\PntSet}{\setC}.   
\end{align*}
\end{problem}
For a given instance of $\FTC(\PntSet,k,\ell)$, we call 
$\COPT$ the optimum solution and we let $\ropt$ denote its cost, i.e. 
$\ropt = \ccost{\PntSet}{\COPT}$. The classical \emphi{$k$-center} clustering problem on a point set $\PntSet$ is $\FTC(\PntSet,k,1)$, and is referred to
as the \emphi{non-fault-tolerant} $k$-center problem.


\begin{problem}[Fault-tolerant $k$-median]
Let $\PntSet$ be set of $n$ points in $\mtrA$, and let $k$ and $\ell$ be two given 
integer parameters such that $1\leq \ell\leq k\leq n$. For a subset 
$\setC\subseteq \PntSet$, we define the cost function $\ccost{\PntSet}{\setC}$ as,
\begin{align*}
    \ccost{\PntSet}{\setC} = \sum_{\pnt \in \PntSet} \distPk{\setC}{\pnt}{\ell}.
\end{align*}
The fault-tolerant $k$-median problem, denoted $\FTM(\PntSet,k,\ell)$, is to compute a set 
$\COPT$ with $\cardin{\COPT} = k$ such that,
\begin{align*}
    \ccost{\PntSet}{\COPT} = \min_{\setC \subseteq \PntSet, \cardin{\setC} = k}
                             \ccost{\PntSet}{\setC}.   
\end{align*}
\end{problem}
For a given instance of $\FTM(\PntSet,k,\ell)$, we call 
$\COPT$ the optimum solution and we let $\sopt$ denote its cost, i.e. 
$\sopt = \ccost{\PntSet}{\COPT}$. The classical \emphi{$k$-median} clustering problem on a point set $\PntSet$ is $\FTM(\PntSet,k,1)$, and is referred to as the \emphi{non-fault-tolerant}
$k$-median problem.


\section{Algorithms}
\seclab{algos}
Our algorithms for both problems, $\FTC(\PntSet,k,\ell)$ and $\FTM(\PntSet,k,\ell)$,
have the same structure. In the first step they run an approximation
algorithm for the non-fault-tolerant version of the respective problem, 
for $\qtnt=\floor{k/\ell}$ centers, 
and in the second step, the solution output by the first step is added to in 
a straightforward manner described below. 
Notice that for either fault-tolerant problem, 
any approximation algorithm for the non-fault-tolerant version can be used in the first step.
In particular, we prove that if the chosen algorithm for this first step is a 
$\const$-approximation algorithm for the non-fault-tolerant problem for $\qtnt$ centers, 
then the set we output at the end of step two will be a 
$(1+2\const)$-approximation (resp. $(1+4\const)$-approximation) for the fault-tolerant 
$k$-center (resp. $k$-median) problem with $k$ centers. 

Natural choices to use for our non-fault-tolerant 
$\qtnt$-median algorithm include the local search algorithm of 
Arya \etal \cite{agkmp-lshkm-01}, which is favored for its combinatorial nature, and 
simplicity of implementation, or the recent algorithm by 
Li and Svensson \cite{ls-akmpa-13}, which facilitates a slight improvement in 
the approximation factor. For the algorithms of Arya \etal and that of Li and Svennson we
refer the reader to the respective papers, as knowledge of these algorithms is not 
required for understanding our algorithm. We let $\kMedian(\PntSet,\qtnt)$ 
denote the chosen approximation algorithm for $\qtnt$-median.

Similarly, we let $\kCenter(\PntSet,\qtnt)$ denote the chosen approximation algorithm for 
non-fault-tolerant $\qtnt$-center. Perhaps the most natural choice for our $\qtnt$-center 
algorithm is the $2$-approximation algorithm by Gonzalez \cite{g-cmmid-85}.  
In fact, in \secref{improved} we show that this particular choice 
leads to a simpler analysis than the general case, and produces a much tighter 
approximation ratio guarantee.  Since knowledge of the algorithm of Gonzalez is needed for 
this analysis, we briefly review this algorithm below in \secref{algo:gonzalez}.

\subsection{Fault-tolerant algorithms}
We now describe the algorithms for fault-tolerant $k$-center and fault-tolerant $k$-median, 
that is $\FTC(\PntSet,k,\ell)$ and $\FTM(\PntSet,k,\ell)$.

For the problem $\FTC(\PntSet,k,\ell)$ (resp. $\FTM(\PntSet,k,\ell)$) 
first run the algorithm $\kCenter(\PntSet,\qtnt)$ (resp. $\kMedian(\PntSet,\qtnt)$). 
Let $\PntSetA \subseteq \PntSet$ denote the set of $\qtnt$ centers output, and let 
$\PntSetA = \{ \pntA_1, \dots, \pntA_\qtnt \}$.
Then the set of centers we output for our fault-tolerant solution is, 
$\setC = \bigcup_{i=1}^\qtnt \NNPk{\PntSet}{\pntA_i}{\ell}$.
That is, we take the $\ell$ nearest neighbors of each point $\pntA_i$ in $\PntSet$,
for $i = 1,\dots, \qtnt$. We only use this set $\setC$ in the analysis. If however $\setC$
has less than $k$ points, we can throw in $k - \cardin{\setC}$ additional points
chosen arbitrarily from $\PntSet \setminus C$, since adding additional centers can only decrease the cost of our solution.

Let $\fCenter(\PntSet,k,\ell)$ and $\fMedian(\PntSet,k,\ell)$ denote these algorithms 
for $\FTC(\PntSet,k,\ell)$ and $\FTM(\PntSet,k,\ell)$, respectively. 


\subsection{The algorithm of Gonzalez}
\seclab{algo:gonzalez}
We now describe the $2$-approximation algorithm for the $\qtnt$-center problem, 
due to Gonzalez \cite{g-cmmid-85}. Gonzalez's algorithm builds a solution set $\setC$ 
iteratively. To kick-start the iteration, we let $\setC = \{ \pnt \}$ where $\pnt \in \PntSet$ is an arbitrary point. Until $\qtnt$ points have been accumulated, the algorithm repeatedly looks 
for the furthest point in $\PntSet$ to the current set $\setC$, and adds the found point 
to $\setC$. More formally, at each step we compute 
$\arg \max_{\pntA \in \PntSet } \dist{\pntA}{\setC}$, 
and add it to $\setC$.

This algorithm is not only simple from a conceptual stand-point, 
but also in regards to implementation and running time.  Indeed, by just maintaining 
for each point in $\PntSet$, its current nearest center among $\setC$, the above 
algorithm can be implemented in $O(n)$ time per iteration, for a total time of $O(n\qtnt)$. 
As mentioned earlier, the result of Hochbaum and Shmoys \cite{hs-bphkc-85} implies that the
approximation factor for this algorithm for general metric spaces, is the best possible.


\section{Results and Analysis}
\seclab{results}

We now present our results and their proofs. Our first result, is that using a factor 
$\const$-approximation algorithm for $\kMedian(\PntSet,\qtnt)$ in the algorithm 
$\fMedian(\PntSet,k,\ell)$ gives a $(1+4\const)$-approximation algorithm for the
problem $\FTM(\PntSet,k,\ell)$.
The structure of the $k$-center problem allows us to use a nearly identical analysis 
except with one simplification, yielding an improved $(1+2\const)$-approximation algorithm 
for the problem $\FTC(\PntSet,k,\ell)$. 
Our second result, 
shows that if one uses the  algorithm of 
Gonzalez \cite{g-cmmid-85} for the subroutine $\kCenter(\PntSet,\qtnt)$, 
then one can guarantee a tighter approximation ratio of $4$ (or $3$ if $l\divides k$), 
as opposed to the $5$ guaranteed by our first result. 
\subsection{Analysis for fault-tolerant $k$-median} 
\begin{theorem}
\thmlab{median}
For a given point set $\PntSet$ in a metric space $\mtrA$ with $\cardin{\PntSet} = n$, 
the algorithm $\fMedian(\PntSet,k,\ell)$ achieves a $(1+4\const)$-approximation to the 
optimal solution of $\FTM(\PntSet,k,\ell)$, 
where $\const$ is the approximation guarantee of the 
subroutine $\kMedian(\PntSet,\qtnt)$, where $\qtnt = \floor{k / \ell}$.
\end{theorem}

As a corollary we have,
\begin{corollary}
  There is a $12$-approximation algorithm for the
  problem $\FTM(\PntSet,k,\ell)$.
\end{corollary}

\begin{proof}
  We use the $(1 + \sqrt{3} + \eps)$-approximation algorithm of 
  Li and Svennson \cite{ls-akmpa-13} with a small enough $\eps$, 
  for the subroutine $\kMedian(\PntSet,\qtnt)$. The result
  follows by appealing to \thmref{median}.
\end{proof}


\subsubsection*{Proof of \thmref{median}}
We refer the reader to \secref{prelims} for notation already introduced. We need some
more notation. For a given instance of $\FTM(\PntSet,k,\ell)$, 
let $\setC^* = \{w_1, w_2, \ldots, w_{k}\}$ be an optimal set of centers, 
and let $\sopt$ be its cost, i.e, 
$\sopt = \sum_{\pnt\in \PntSet} \distPk{\setC^*}{\pnt}{\ell}$.  Let 
$\setC = \{c_1, \dots, c_k\}$ be the set of centers returned by 
$\fMedian(\PntSet,k,\ell)$, and $\salg$ its cost. 

Let $\qtnt = \floor{k/\ell}$, and let $\smed$ denote the cost of the 
optimum $\qtnt$-median clustering on $\PntSet$, i.e., the optimum for the 
problem $\FTM(\PntSet,\qtnt,1)$.  When $\fMedian(\PntSet,k,\ell)$ is run, 
it makes a subroutine call to $\kMedian(\PntSet, \qtnt)$.  
Let $\PntSetA = \{\pntA_1,\dots, \pntA_{\qtnt}\}$ be the set of centers returned 
by this subroutine call.  We know that $\PntSetA$ is a 
$\const$-approximation to the optimal solution to $\FTC(\PntSet,\qtnt, 1)$. 

Notice that, $\setC$ includes $\bigcup_{i=1}^\qtnt \NNPk{\PntSet}{\pntA_i}{\ell}$. 
We assume that the set $\setC$ has exactly
$k$ points. As mentioned earlier, we only require that $\setC$ includes 
$\bigcup_{i=1}^\qtnt \NNPk{\PntSet}{\pntA_i}{\ell}$ in our analysis,
and if $\cardin{\bigcup_{i=1}^\qtnt \NNPk{\PntSet}{\pntA_i}{\ell}} < k$,
we can always add additional points. This can only decrease the cost of
clustering. 

Proving the following two claims will immediately imply $\salg \leq (1+4\const)\sopt$.

\begin{claim}
\claimlab{first}
 We have that, $\salg \leq \sopt + 2\const\smed$.
\end{claim}

\begin{claim}
\claimlab{second}
 We have that, $\smed \leq 2\sopt$.
\end{claim}

\begin{proofof}{\claimref{first}}
Let $\pnt \in \PntSet$, and let $\pntA = \nnPk{\PntSetA}{\pnt}{1}$. 
By \obsref{lipschitz}, 
$\distPk{\setC}{\pnt}{\ell} \leq \dist{\pnt}{\pntA} + \distPk{\setC}{\pntA}{\ell}$.
As $\NNPk{\PntSet}{\pntA}{\ell} \subseteq \setC$, we have that
$\distPk{\PntSet}{\pntA}{\ell} = \distPk{\setC}{\pntA}{\ell}$.
Again by \obsref{lipschitz}, 
$\distPk{\PntSet}{\pntA}{\ell} \leq \dist{\pntA}{\pnt} + \distPk{\PntSet}{\pnt}{\ell}$. 
Combining the two inequalities gives,
$\distPk{\setC}{\pnt}{\ell} \leq 2\dist{\pnt}{\pntA}+\distPk{\PntSet}{\pnt}{\ell}
= 2\distPk{\PntSetA}{\pnt}{1}+ \distPk{\PntSet}{\pnt}{\ell}$.
Thus,
\ConfVer{
\begin{align}
   \eqlab{eq1}
   \begin{split}
     \salg &= \sum_{\pnt\in \PntSet} \distPk{\setC}{\pnt}{\ell}\leq \sum_{\pnt\in \PntSet} \pth{2\distPk{\PntSetA}{\pnt}{1}+ 
     \distPk{\PntSet}{\pnt}{\ell}} \\
       &\leq 2c\smed+ \sopt,
   \end{split}
\end{align}
}
\FullVer{
\begin{align}
   \eqlab{eq1}
     \salg = \sum_{\pnt\in \PntSet} \distPk{\setC}{\pnt}{\ell}\leq \sum_{\pnt\in \PntSet} \pth{2\distPk{\PntSetA}{\pnt}{1}+ 
     \distPk{\PntSet}{\pnt}{\ell}}
       \leq 2c\smed+ \sopt,
\end{align}
}
as $\PntSetA$ is a $\const$-approximate $\qtnt$-median solution,
$\distPk{\PntSet}{\pnt}{\ell} \leq \distPk{\setC^*}{\pnt}{\ell}$, and
$\sopt = \sum_{\pnt \in \PntSet} \distPk{\setC^*}{\pnt}{\ell}$. 
\end{proofof}

The following is required to prove \claimref{second}, but is interesting in its own right.
\begin{lemma}
\lemlab{subset}
Let $\mtrA$ be any metric space. Let 
$\PntSetX \subseteq \mtrA$ with $\cardin{\PntSetX} = t$. Then
for any integer $1 \leq h \leq t$, and any finite set
$\PntSetY \subseteq \mtrA$,
there exists a subset $\PntSetB \subseteq \PntSetY$, such that 
\begin{inparaenum}[(A)]
 \item
 $\cardin{\PntSetB} \leq t/h$, and,
 \item
 $\forall \pntY \in \PntSetY$, $\distPk{\PntSetB}{\pntY}{1} \leq 2\distPk{\PntSetX}{\pntY}{h}$.
\end{inparaenum}
\end{lemma}
\begin{proof}
We give an algorithm to construct such a subset $\PntSetB \subseteq \PntSetY$.  This subset 
is constructed by iteratively scooping out the points of the minimum radius ball containing 
$h$ points from $\PntSetX$, adding the center to $\PntSetB$, and repeating.  Formally, let 
$\PntSetC_0 = \emptyset$, and for $i=1,\dots,\floor{t/h}$, define iteratively, 
$\PntSetX_i = \PntSetX \setminus \pth{\bigcup_{j=0}^{i-1} \PntSetC_j}$, 
$\pntY_i = \arg \min_{\pntB \in \PntSetY} \distPk{\PntSetX_i}{\pntB}{h}$, and, 
$\PntSetC_i = \NNPk{\PntSetX_i}{\pntY_i}{h}$.
We prove that $\PntSetB = \bigcup_{i=1}^{\floor{t/h}} \brc{\pntY_i}$, 
is the desired subset of points.
 
First, clearly $\cardin{\PntSetB}\leq t/h$. Let $\pntY \in \PntSetY$, 
and let $\ballA =\ball{\pntY}{\num}$, where $\num=\distPk{\PntSetX}{\pntY}{h}$.  Let 
$\PntSetC_i$ be the first subset, i.e. the one with smallest index $i$, 
such that there exists some point $\pntC \in \ballA \cap \PntSetC_i$. Such a point must exist, 
since fewer than $h$ points are in 
$\PntSetX \setminus \pth{\bigcup_{j=1}^{\floor{t/h}} \PntSetC_j}$, while
$\cardin{\ballA \cap \PntSetX} \geq h$. Clearly 
$\ballA \cap \PntSetX \subseteq \PntSetX_i$, as $i$ is the minimum index such that 
$\ballA \cap \PntSetC_i \neq \emptyset$. As such we have, 
$\distPk{\PntSetX}{\pntY}{h} = \distPk{\PntSetX_i}{\pntY}{h}$. 
Let $r_i = \distPk{\PntSetX_i}{\pntY_i}{h}$, 
be the radius of the ball that scooped out $\PntSetC_i$.  Clearly $r_i \leq \num$, as
\[
  \num 
  = 
  \distPk{\PntSetX}{\pntY}{h} 
  =
  \distPk{\PntSetX_i}{\pntY}{h} 
  \geq 
  r_i 
  = 
  \arg \min_{\pntB \in \PntSetY} \distPk{\PntSetX_i}{\pntB}{h}.
\]
Now, since $\pntC \in \ballA \cap \PntSetC_i$, 
$\dist{\pntY_i}{\pntC} \leq r_i = \distPk{\PntSetX_i}{\pntY_i}{h}$. 
By the triangle inequality,
\FullVer{
 \[
 \distPk{\PntSetB}{\pntY}{1}\leq \dist{\pntY}{\pntY_i} \leq \dist{\pntY}{\pntC} + \dist{\pntC}{\pntY_i} \leq \num +r_i \leq 2\num = 2\distPk{\PntSetX}{\pntY}{h}.
 \]
}
\ConfVer{
  \begin{align*}
    \distPk{\PntSetB}{\pntY}{1} & \leq \dist{\pntY}{\pntY_i} 
    \leq \dist{\pntY}{\pntC} + \dist{\pntC}{\pntY_i} \leq \num +r_i 
    \leq 2\num \\
    &= 2\distPk{\PntSetX}{\pntY}{h}.
  \end{align*}
}
 \end{proof}
 
\begin{proofof}{\claimref{second}}
  We use \lemref{subset} with $\PntSetY = \PntSet$, 
  $\PntSetX = \setC^*$, $t = \cardin{\setC^*} = k$
  and $h = \ell$. 
  Let $\PntSetB$ be the subset of $\PntSet$ guaranteed by \lemref{subset}. 
  Now $\cardin{\PntSetB} \leq k / \ell$, and as
  such $\cardin{\PntSetB} \leq \qtnt$.
  We have,
  \begin{align}
    \eqlab{eq2}
    \smed 
    \leq 
    \sum_{\pnt\in \PntSet} \distPk{\PntSetB}{\pnt}{1} 
    \leq 
    \sum_{\pnt\in \PntSet} 2\distPk{\setC^*}{\pnt}{\ell}
    = 
    2\sopt.
  \end{align}

  The first inequality follows since 
  $\smed$ is the cost of the optimum $\qtnt$-median clustering of $\PntSet$, 
  while $\sum_{\pnt \in \PntSet} \distPk{\PntSetB}{\pnt}{1}$ is the cost of
  a $\cardin{\PntSetB}$-median clustering of $\PntSet$ by the set of centers
  $\PntSetB \subseteq \PntSet$ with $\cardin{\PntSetB} \leq \qtnt$. 
  The second inequality follows from \lemref{subset}.
 \end{proofof}
 This concludes the proof of \thmref{median}.

\subsection{Analysis for fault-tolerant $k$-center}

We now present the analogues result to \thmref{median} 
for fault-tolerant $k$-center.
By following the proof nearly verbatim from the previous section one sees that similar to 
$\fMedian(\PntSet,k,\ell)$, $\fCenter(\PntSet,k,\ell)$ also provides a $(1+4\const)$-approximation.
However, in this case we will actually get a $(1+2\const)$-approximation, since now an improved 
and simpler version of \claimref{first} holds.

As a quick note on notation, here $\ralg$, $\ropt$, and $\rcen$ will play the analogues role for center 
as $\salg$, $\sopt$, and $\smed$ played for median.

\begin{claim}
\claimlab{firstImproved}
 We have that, $\ralg \leq \ropt + 2\const\rcen$.
\end{claim}

\begin{proofof}{\claimref{firstImproved}}
Let $\pnt \in \PntSet$, and let $\pntA = \nnPk{\PntSetA}{\pnt}{1}$. 
By \obsref{lipschitz}, 
$\distPk{\setC}{\pnt}{\ell} \leq \dist{\pnt}{\pntA} + \distPk{\setC}{\pntA}{\ell}
= \dist{\pnt}{\pntA} + \distPk{\PntSet}{\pntA}{\ell}$, where the equality follows since 
$\NNPk{\PntSet}{\pntA}{\ell} \subseteq \setC$.
Thus,
\begin{align}
   \eqlab{eq3}
   \begin{split}
     \ralg &= \max_{\pnt\in \PntSet} \distPk{\setC}{\pnt}{\ell}\leq \max_{\pnt\in \PntSet} \pth{\distPk{\PntSetA}{\pnt}{1}+ 
     \distPk{\PntSet}{\pntA}{\ell}} \\
       &\leq 2c\rcen+ \ropt,
   \end{split}
\end{align}
as $\PntSetA$ is a $\const$-approximate $\qtnt$-center solution,
$\distPk{\PntSet}{\pntA}{\ell} \leq \distPk{\setC^*}{\pntA}{\ell}$, and
$\ropt = \max_{\pnt \in \PntSet} \distPk{\setC^*}{\pnt}{\ell}$. 
\end{proofof}

\begin{theorem}
\thmlab{center:main}
  For a given point set $\PntSet$ in a metric space $\mtrA$ with $\cardin{\PntSet} = n$, 
  the algorithm $\fCenter(\PntSet,k,\ell)$ achieves a $(1+2\const)$-approximation to the 
  optimal solution of $\FTC(\PntSet,k,\ell)$, 
  where $\const$ is the approximation guarantee of the 
  subroutine $\kCenter(\PntSet,\qtnt)$, where $\qtnt = \floor{k / \ell}$.
\end{theorem}

\begin{proof}
  As stated above,
  the proof of this theorem is very similar to the proof of \thmref{median}. In fact,
  we can repeat the proof of \thmref{median} almost word for word, except that we need
  to replace the sum function $\sum \limits_{\pnt \in \PntSet}$ by the max function,
  $\max \limits_{\pnt \in \PntSet}$. More specifically, this needs to be done 
  for \Eqref{eq2} in the proof of \claimref{second}, and to replace \Eqref{eq1} from \claimref{first} we 
  instead use the improved \Eqref{eq3} from \claimref{firstImproved}. 
  As the proof can be reconstructed step-by-step from the
  detailed proof of \thmref{median} by making these modifications, we omit
  it for the sake of brevity.
\end{proof}

\subsubsection{A tighter analysis when using Gonzalez's algorithm as a subroutine}
\seclab{improved}
If we use a $2$-approximation algorithm for the subroutine $\kCenter(\PntSet,\qtnt)$, \thmref{center:main} 
implies that $\fCenter(\PntSet,k,\ell)$ is a $9$-approximation algorithm. Here we present a 
tighter analysis for the case when we use the $2$-approximation algorithm of Gonzalez \cite{g-cmmid-85}
(see also \secref{algo:gonzalez}) 
for the subroutine $\kCenter(\PntSet,\qtnt)$.

See \secref{prelims} for definitions and notation introduced previously. 
Some more notation is needed.
Let $\setC^* = \{w_1, w_2, \ldots, w_{k}\}$ be an optimal set of centers. Its cost, $\ropt$, is 
$\max_{\pnt\in \PntSet} \distPk{\setC^*}{\pnt}{\ell}$.  Let 
$\setC = \{c_1, \dots, c_k\}$ be the set of centers returned by $\fCenter(\PntSet,k,\ell)$, 
and let $\ralg$ be its cost.

Let $\qtnt = \floor{k/\ell}$, where for now we assume $\ell\divides k$, i.e, $\qtnt = k / \ell$.
As we show later, this assumption can be removed. When $\fCenter(\PntSet,k,\ell)$ is run, 
it makes a subroutine call to $\kCenter(\PntSet, \qtnt)$. As mentioned, in this section 
we require this subroutine to be the algorithm of Gonzalez \cite{g-cmmid-85}. 
Let $\PntSetA = \{\pntA_1,\dots, \pntA_{\qtnt}\}$ be the set of centers returned by this 
subroutine call. Additionally, let 
$r_i = \dist{\pntA_i}{\PntSetA_{i-1}}$ for $2 \leq i \leq \qtnt$, 
where $\PntSetA_{i-1} = \{\pntA_1,\dots, \pntA_{i-1}\}$. We assume $m > 1$, 
as the $m=1$ case is easier.

The following is easy to see, and is used in the correctness proof for the algorithm
of Gonzalez. See \cite{h-gaa-11} for an exposition.
\begin{lemma}
\lemlab{gonz}
For $i\neq j$, $\dist{\pntA_i}{\pntA_j} \geq r_\qtnt$. 
\end{lemma}

\begin{lemma}
\lemlab{easy}
  For any $\pntA_i$, $\NNPk{\setC^*}{\pntA_i}{\ell} \subseteq \ball{\pntA_i}{\ropt}$ and
  $\NNPk{\PntSet}{\pntA_i}{\ell} \subseteq \ball{\pntA_i}{\ropt}$.
\end{lemma}
\begin{proof}
The first claim follows since $\pntA_i\in \PntSet$ and so 
$\distPk{\setC^*}{\pntA_i}{\ell} \leq \ropt$. As $\setC^* \subseteq \PntSet$, the second claim 
follows.
\end{proof}

\begin{lemma}
\lemlab{near}
We have that, $\ralg \leq r_\qtnt+\ropt$.
\end{lemma}
\begin{proof}
As in Gonzalez's algorithm, we have 
$r_\qtnt = \max_{\pnt\in \PntSet} \dist{\pnt}{\PntSetA_{\qtnt-1}}$, 
and so $\dist{\pnt}{\PntSetA} \leq r_\qtnt$ for any $\pnt \in \PntSet$. 
Consider any point $\pnt\in \PntSet$, and 
let $\pntA = \nnPk{\PntSetA}{\pnt}{1}$.  
By how $\fCenter(\PntSet,k,\ell)$ is defined, $\NNPk{\PntSet}{\pntA}{\ell}\subseteq \setC$,
and so $\distPk{\setC}{\pntA}{\ell} = \distPk{\PntSet}{\pntA}{\ell} 
\leq \distPk{\setC^*}{\pntA}{\ell} \leq \ropt$.
By \obsref{lipschitz} we have, 
$\distPk{\setC}{\pnt}{\ell} \leq \dist{\pnt}{\pntA} + \distPk{\setC}{\pntA}{\ell}
\leq r_\qtnt + \ropt$.
\end{proof}

\begin{lemma}
\lemlab{unique}
If $\ralg > 3\ropt$, then for any $1 \leq i\neq j \leq \qtnt$, 
$\ball{\pntA_i}{\ropt}$ and $\ball{\pntA_j}{\ropt}$ are disjoint and 
each contains at least $\ell$ centers from $\setC^*$.
\end{lemma}
\begin{proof}
Let $\pntA_i$ and $\pntA_j$ be any two distinct centers in $\PntSetA$.  
By \lemref{gonz} and \lemref{near}, 
$\dist{\pntA_i}{\pntA_j}\geq r_\qtnt \geq \ralg-\ropt > 2\ropt$, which implies that,
$\ball{\pntA_i}{\ropt} \cap \ball{\pntA_j}{\ropt} = \emptyset$. 
Each ball contains $\ell$ centers from $\setC^*$ by \lemref{easy}. 
\end{proof}

\begin{lemma}
\lemlab{final}
We have that, $\ralg\leq 3\ropt$.  
\end{lemma}
\begin{proof}
Suppose otherwise that $\ralg > 3\ropt$.  
By \lemref{unique}, for $i = 1,\dots,\qtnt$, 
$\cardin{\ball{\pntA_i}{\ropt} \cap \setC^*} \geq \ell$, and for $1 \leq i < j \leq \qtnt$,
$\ball{\pntA_i}{\ropt} \cap \ball{\pntA_j}{\ropt} = \emptyset$. Assign all points in 
$\setC^* \cap \ball{\pntA_i}{\ropt}$ to $\pntA_i$. Notice, $\pntA_i$ is the unique point
from $\PntSetA$ within distance $\ropt$ for any point assigned to it.  
Now $\cardin{\PntSetA} = \qtnt = k/\ell$, and each point in $\PntSetA$ 
gets at least $\ell$ points of 
$\setC^*$ assigned to it uniquely. As such, there are at least $\qtnt \ell = k$ points of 
$\setC^*$ assigned to some point of $\PntSetA$. Since $\cardin{\setC^*}= k $, it follows that 
each center in $\setC^*$ gets assigned to a center in $\PntSetA$ within distance $\ropt$. 
For $\pnt \in \PntSet$, let $\pntB$ be its closest center in $\setC^*$.  Let
$\pntA$ be $\pntB$'s center from $\PntSetA$ in distance $\leq \ropt$. We have 
$
\dist{\pnt}{\pntA} \leq \dist{\pnt}{\pntB}+ \dist{\pntB}{\pntA} \leq \ropt + \ropt = 2\ropt,
$
by the triangle inequality. As $\NNPk{\PntSet}{\pntA}{\ell} \subseteq \setC$, we have that
$\distPk{\setC}{\pntA}{\ell} = \distPk{\PntSet}{\pntA}{\ell} \leq \distPk{\setC^*}{\pntA}{\ell}
\leq \ropt$. By \obsref{lipschitz}, we have that, 
$\distPk{\setC}{\pnt}{\ell} \leq \distPk{\setC}{\pntA}{\ell} + \dist{\pnt}{\pntA} 
\leq \ropt + 2\ropt = 3\ropt$. This implies $\ralg \leq 3 \ropt$, a contradiction.
\end{proof}

\begin{theorem}
\thmlab{center}
For a given instance of $\FTC(\PntSet,k,\ell)$, when using the 
algorithm of Gonzalez \cite{g-cmmid-85} for the subroutine $\kCenter(\PntSet,\qtnt)$, 
the algorithm $\fCenter(\PntSet,k,\ell)$ achieves a 
$4$-approximation to the optimal solution to $\FTC(\PntSet,k,\ell)$, 
and a 3-approximation when $\ell\divides k$. 
\end{theorem}

\begin{proof}
  The $\ell \divides k$ case follows from \lemref{final}. If $\ell$ does not divide $k$,
  the proof of \lemref{final} needs to be changed as follows. Suppose, $k = \ell*\qtnt + r$
  for some integer $0 < r < \ell$. Let $k' = \ell * \qtnt$.
  As in the proof of \lemref{final}, it follows from
  \lemref{unique}, that if $\ralg > 3\ropt$, then at least $k'$ centers from
  $\setC^*$ will be within distance at most $\ropt$ to a center in $\PntSetA$. Therefore,
  there are at most $k - k' = r$ centers from $\setC^*$, that are not within $\ropt$ to some
  point in $\PntSetA$. However, each such center needs $\ell > r$ centers from $\setC^*$,
  to be within distance $\ropt$, and so each such center must be within distance $\ropt$ from
  one of the centers of $\setC^*$ that is near a center in $\PntSetA$, i.e. within distance
  $\ropt$ to some center in $\PntSetA$. Hence, by the triangle inequality, each center in
  $\setC^*$, has a center of $\PntSetA$ within distance at most $2 \ropt$. Repeating the 
  argument of \lemref{final}, with this different upper bound, we get that $\ralg \leq 4 \ropt$.
\end{proof}


\section{Conclusions}
\seclab{conclusions}

In this paper we investigated fault-tolerant variants of the $k$-center and $k$-median 
clustering problems. Our algorithm achieves a $(1 + 2\const)$-approximation 
(resp. $(1 + 4\const)$-approximation) factor, where 
$\const$ is the approximation factor for the non-fault-tolerant 
$\qtnt$-center (resp. $\qtnt$-median) algorithm that 
we use as a subroutine. Using a better analysis for the case
of fault-tolerant $k$-center, when we use Gonzalez's algorithm as a subroutine, we 
showed that our algorithm has a tighter approximation ratio of $4$. For 
fault-tolerant $k$-median, we get a $(5 + 4\sqrt{3} + \eps) \approx 12$-approximation 
algorithm, by using  the recent algorithm of 
Li and Svensson as a subroutine \cite{ls-akmpa-13}. We can see several questions 
for future research.
\begin{itemize}
  \item The best known approximation factor for the fault-tolerant $k$-center problem is
        $2$ by Chaudhuri \etal \cite{cgr-pnkcp-98} and Khuller \etal \cite{kps-ftkcp-00}. 
        Their techniques are based on the work of Hochbaum and Shmoys \cite{hs-bphkc-85} and
        Krumke \cite{k-gpcp-95}. Our algorithm, which
        leads to a $4$-approximation for fault-tolerant $k$-center is
        based on the $2$-approximation to $k$-center by Gonzalez \cite{g-cmmid-85}. Can the 
        algorithm or its analysis be improved to get a factor $2$-approximation? 
        Also, can we deal with the second variant of fault-tolerant 
        $k$-center in the work of Khuller \etal -- which also happens to be the version 
        considered by Krumke and Chaudhuri \etal?

  \item The fault-tolerant $k$-median variant that we investigate, is very different 
        from the work
        of Swamy and Shmoys \cite{ss-ftfl-03}, but their techniques are more 
        technically involved.
        As we show, we reduce the fault-tolerant version to the non-fault-tolerant version
        for a smaller number of centers. An 
        important question that arises is the following: Can the version considered by
        Swamy and Shmoys be reduced to the non-fault-tolerant version, or some 
        variant thereof, i.e., can we use some simpler problem as an oracle to get a
        fault-tolerant $k$-median algorithm, for the version of Swamy and Shmoys?
\end{itemize}

\section*{Acknowledgements}
We would like to thank Sariel Har-Peled for useful discussions, and in particular for
the discussion that led us to think about this problem.

\ConfVer{
  \bibliographystyle{abbrv}
}
\FullVer{
\bibliographystyle{plain}%
}
\bibliography{lnnkc}%

\end{document}